\documentclass[a4paper,11pt]{article}

\usepackage{amsfonts,amsmath,amssymb,amsthm,epsfig,psfrag}
\usepackage[all]{xy}

\usepackage{psfrag, epsfig, amsmath,amssymb,amsthm,accents}

\newtheorem{proposition}{Proposition}
\newtheorem{theorem}{Theorem}

\newtheorem{lemma}{Lemma}
\newtheorem{corollary}{Corollary}

\newtheorem{example}{Example}

\def\N{{\mathbb N}}

\def\ind{{\mathbf{1}}}
\def\ore{\overrightarrow{e}}
\def\orE{\overrightarrow{E}}

\def\d{\partial}
\def\bx{\bold{x}}
\def\by{\bold{y}}
\def\bb{\bold{b}}

\def\bB{\bold{B}}
\def\bC{\bold{C}}

\def\Pcal{{\mathcal P}}
\def\Rcal{{\mathcal R}}
\def\Dcal{{\mathcal D}}
\def\Qcal{{\mathcal Q}}

\def\bY{\bold{Y}}
\def\bI{\bold{I}}
\def\bJ{\bold{J}}
\def\bK{\bold{K}}
\def\bZ{\bold{Z}}
\def\bW{\bold{W}}

\def\bX{\bold{X}}

\def\bac{\backslash}

\def\ind{{\mathbf{1}}}

\def\R{{\mathbb R}}

\def\N{{\mathbb N}}

\newcommand{\BEAS}{\begin{eqnarray*}}
\newcommand{\EEAS}{\end{eqnarray*}}
\newcommand{\BEA}{\begin{eqnarray}}
\newcommand{\EEA}{\end{eqnarray}}
\newcommand{\BEQ}{\begin{equation}}
\newcommand{\EEQ}{\end{equation}}
\newcommand{\BIT}{\begin{itemize}}
\newcommand{\EIT}{\end{itemize}}
\newcommand{\BNUM}{\begin{enumerate}}
\newcommand{\ENUM}{\end{enumerate}}

\begin{document}
%
\title{Loopy Annealing Belief Propagation for vertex cover and matching: convergence, LP relaxation, correctness and Bethe approximation}

\author{M. Lelarge\footnote{INRIA-ENS, Paris, France,
email: marc.lelarge@ens.fr}}
\date{}

\maketitle

\begin{abstract}
For the minimum cardinality vertex cover and maximum cardinality
matching problems, the max-product form of belief propagation (BP) is
known to perform poorly on general graphs. In this paper, we present an iterative loopy annealing BP
(LABP) algorithm which is shown to converge and to solve a Linear Programming
relaxation of the vertex cover or matching problem on general graphs. 
LABP finds (asymptotically) a minimum half-integral vertex cover
(hence provides a 2-approximation)  and a maximum fractional matching
on any graph.
We also show that LABP finds (asymptotically) a minimum size vertex
cover for any bipartite graph and as a consequence compute the
matching number of the graph. Our proof relies on some subtle
monotonicity arguments for the local iteration.
We also show that the Bethe free entropy is concave and that LABP
maximizes it. Using loop calculus, we also give an exact (also
intractable for general graphs) expression of the partition function
for matching in term of the LABP messages which can be used to improve
mean-field approximations. 
\end{abstract}


\section{Introduction}

Belief propagation was originally formulated by Judea Pearl \cite{pearl} as a
distributed algorithm to perform statistical inference. Computing
marginals is, in general, an expensive operation. If the probability
distribution can be written as a product of factors that only
depend on a small subset of the variables, then one could possibly
compute the marginals much faster. This decomposition is captured by a
corresponding graphical model. When the graphical model is a tree, the
belief propagation (BP) algorithm is guaranteed to converge to the exact
marginals. If the algorithm is run on an arbitrary graph that is not a
tree, then neither convergence nor correctness are guaranteed.

Loopy BP algorithms 
\cite{yfw05} have been shown empirically to be effective in solving a
wide range of hard problems in various fields \cite{wj08},
\cite{ru01}. 
Understanding their convergence and
accuracy on general graphs remains an active research area.
There is a large literature on max-product BP algorithm, a variant of the BP
algorithm, computing the assignment of the variables that maximizes a
given objective function. However, rigorously characterizing their
behavior has proved challenging. 
We note that there is a vast literature trying to resolve the difficulties raised by max-product BP
algorithm for general models.
Several alternate message passing schemes have been
proposed: MPLP \cite{gj07}, tree-reweighted max-product (TRMP)
\cite{wjw05} or max-sum diffusion (MSD) \cite{Werner}.
As shown in \cite{mgw09}, these algorithms provide bounds for the
log-partition function and, under a suitable schedules of the updates,
are guaranteed to converge to the optimum solution if there is a
unique optimum.
In this paper, we consider the vertex
cover and matching problems and propose a new line of research leading
to simple parallel BP algorithms for these problems. 

We first review some theoretical works about the performances of BP
algorithms for combinatorial optimization problems including matching, independent set and network flow.
In \cite{bss08}, the max-product BP algorithm is shown to find in pseudo-polynomial
time a maximum weight matching in a bipartite graph provided that the
optimal matching is unique. \cite{bbcr11} and \cite{smw11}
generalize this result by establishing convergence and correctness of
the max-product BP when the LP relaxation has a unique optimum and
this optimum is integral. \cite{smw11} also shows that when this
condition is not satisfied then max-product BP will give useless
estimates for some edges.
By setting all the weights to one, the results of \cite{bbcr11}, \cite{smw11} apply to our setting of maximum cardinality matching:
max-product BP converges and is correct only when the graph has a unique maximum matching which is optimum for the LP problem.

For the vertex cover problem, a one-sided relation between LP relaxation and BP is established: 
\cite{ssw09} shows that for the maximum weight independent set
problem, if the max-product BP algorithm (started from the natural
initial condition) converges then it is correct and the LP problem has a unique
integral solution. 
Since a subset of vertices is a vertex cover if and only if its
complement is an independent set, by setting all the weights to one,
results of \cite{ssw09} apply to the minimum cardinality vertex
cover: the tightness of the LP relaxation is necessary
for the max-product optimality but it is not sufficient.

We should stress that \cite{bbcr11} and \cite{smw11} deal with a
generalization of the matching problem namely with b-matchings. Also
\cite{gsw12} extends \cite{bss08} and analyzes the max-product BP
applied to the minimum-cost network flow problem. In \cite{bss08},
\cite{bbcr11}, \cite{smw11}, \cite{ssw09} or \cite{gsw12}, a crucial
assumption is required for convergence and correctness of BP:
uniqueness of the optimum solution. 
For the minimum cardinality vertex cover and maximum cardinality
matching problems, this assumption is very restrictive.

In this paper, we will overcome this difficulty by a different
approach: we introduce annealing, i.e. we study a relaxed version of
the optimization problem parametrized by the parameter $z>0$
(sometimes called the inverse temperature) such
that in the limit $z\to \infty$, we recover the original optimization
problem.
Since the work of Heilmann and Lieb \cite{heilmannlieb}, it is well
known that the matching problem has the 'correlation decay' property
at positive temperature, i.e. for $z<\infty$. Building on
\cite{bls12}, \cite{lel}, \cite{salez12}, we can show that this
property ensures the convergence of our BP algorithm on any finite
graph as long as $z<\infty$. However, we are only interested in the
limit $z\to \infty$. Our
first main contribution shows that our BP algorithm computes in the limit $z\to
\infty$ a maximum fractional matching.
For the minimum vertex cover, this approach seems doomed to fail. It
is well-known that there is no 'correlation
decay' at low temperature (i.e. as $z\to \infty$) for the independent set
problem making it very hard or even impossible to relate the minimum
vertex cover and its relaxed version with $z<\infty$.
However, the fractional vertex
cover problem is the dual of the fractional matching problem.
As a consequence, we will show our
second main contribution: our BP algorithm computes in the limit $z\to
\infty$ a minimum half-integral vertex cover, hence providing a
2-approximation. 
Surprisingly, if the graph is bipartite, it computes a minimum vertex cover.

To the best of our knowledge, our results are the first rigorous
results in the regime $z\to\infty$ on arbitrary graphs showing the
performances of BP algorithm.
A similar approach based on the Bethe approximation was proposed in \cite{che08} but no convergence
results were given. As noted in \cite{wym12}, \cite{gsc13}, if
convergence is proved, then it is easy to see that our BP algorithm
solves the LP relaxation of the combinatorial optimization problem.
Our paper shows rigorously that this approach is
successful for the minimum vertex cover problem and the maximum
matching problem. 
Also related to our approach is \cite{von13} which
deals with the sum of weighted perfect matchings in complete bipartite
graphs and shows that the Bethe free entropy is concave which easily
implies that the same is true in our setting. 
We note that the Bethe approximation is only used in the analysis of
BP for the maximum matching problem and not for the minimum vertex
cover. The main technical contribution of this paper is in the
analysis of the minimum vertex cover problem. As explained above, a naive direct
approach studying BP for vertex cover fails. Instead we made a careful
analysis of the BP algorithm for the dual relaxed problem (namely
fractional matching) in order to be able to get results for the
original vertex cover problem. This analysis requires original
techniques based on subtle monotonicity arguments for the local iteration.

We end this introduction by a last motivation for this work.
The increasing need to reason about large-scale graph-structured data
in machine learning and data mining has driven the development of new
graph-parallel abstractions such as Pregel \cite{pregel}, Graphlab
\cite{graphlab} and Powergraph \cite{powergraph} that encode
computation as vertex-programs which run in parallel and interact
along edges in the graph. In this setting, BP algorithms present
several opportunities for parallelism: given the messages from the
previous iteration, each new message can be computed completely
independently and in any order. BP algorithms are
certainly natural candidates to leverage the performance and
scalability of graph-parallel abstractions.
Recent parallel implementations of BP
\cite{glg09} show promising empirical results in this direction and
our work makes a significant step towards a better understanding of BP
algorithms that could be extended to other optimization problems.

We present our results in the next Section. We first start in Section
\ref{sec:fmvc} by introducing the two combinatorial optimization
problems studied in this paper: matching and vertex cover. We
introduce our annealing BP and show its convergence for general
graphs. We then introduce a simpler version of BP (corresponding to
the standard max-product version) and relate it to our annealing
BP. We show that it allows us to compute minimum fractional vertex
cover for any graph. In Section \ref{sec:bi}, we show that for
bipartite graphs, our algorithm computes a minimum vertex cover. In
Section \ref{sec:pos}, we use variational techniques to analyze BP and
give exact loop series expansion as developed in \cite{cc06}. We
conclude in Section \ref{sec:conc}. In the Appendix \ref{sec:proofs}, we
provide the detailed proofs.

\section{Results}\label{sec:finite}

\subsection{(Fractional) matching and vertex cover numbers}\label{sec:fmvc}
We consider a graph $G=(V,E)$. We denote by the same symbol $\d v$ the set of
neighbors of node $v\in V$ and the set of edges incident to $v$.
A matching is encoded by a binary vector
$\bB=(B_e,\: e\in E)\in \{0,1\}^E$ defined by $B_e=1$ if and only if
the edge $e$ belongs to the matching.
We have for all $v\in V$, $\sum_{e\in \d v}B_e\leq 1$.
The size of the matching is given by $\sum_e B_e$.
For a finite graph $G$, we define the matching number of $G$ as
$\nu(G)=\max\{\sum_e B_e\}$ where the maximum is taken over matchings of
$G$.
Similarly a vertex cover is encoded by a binary vector
$\bC=(C_v,\:v\in V)\in \{0,1\}^V$ defined by $C_v=1$ if and only if
the vertex $v$ belongs to the vertex cover. We have for all $e=(uv)\in
E$, $C_u+C_v\geq 1$. The size of the vertex cover is given by $\sum_v
C_v$ and the vertex cover number of $G$ is $\tau(G)=\min\{\sum_v
C_v\}$ where the minimum is taken over vertex covers of $G$.

The matching number is the solution of the following binary Integer
Linear Program (ILP):
\begin{eqnarray*}
\nu(G) &=&\max \sum_{e\in E} x_e\\
&&\mbox{s.t.} \quad \sum_{e\in \d v} x_e\leq 1,\: \forall v\in V;\: x_e\in \{0,1\},
\end{eqnarray*}
and the vertex cover number is the solution of the following ILP:
\begin{eqnarray*}
\tau(G) &=&\min \sum_{v\in V} y_v\\
&&\mbox{s.t.} \quad y_u+y_v\geq 1,\: \forall (uv)\in E;\: y_v\in \{0,1\},
\end{eqnarray*}
The straightforward Linear Programming (LP) relaxation of these ILP is
formed by replacing $x_e\in \{0,1\}$ (resp. $y_v\in \{0,1\}$) by
$x_e\in [0,1]$ (resp. $y_v\in [0,1]$).

We define the fractional matching polytope:
\begin{eqnarray}
\label{def:LG} FM(G) = \left\{ \bx \in \R^E,\: x_e\geq 0,\:\sum_{e\in \d v} x_e\leq 1\right\},
\end{eqnarray}
and the fractional matching number
\begin{eqnarray}
\label{def:nu*}\nu(G)\leq \nu^*(G) = \max_{\bx \in FM(G)} \sum_{e\in E}x_e.
\end{eqnarray}
Similarly, we define the fractional vertex cover polytope:
\BEA
FVC(G) = \left\{ \by \in \R^V,\: 0\leq y_v\leq1,\: y_u+y_v\geq 1,\: \forall (uv)\in E \right\},
\EEA
and the fractional vertex cover number is 
\BEA
\label{def:tau*}\tau(G)\geq \tau^*(G) =
\min_{\by\in FVC(G)}\sum_v
y_v.
\EEA
 By linear programming duality, we have $\tau^*(G) = \nu^*(G)$
(see Section 64.6 in \cite{sch03}). Note however that computing the
matching number $\nu(G)$ can be done in polynomial time whereas
determining the vertex cover number $\tau(G)$ is NP-complete
(Corollary 64.1a in \cite{sch03}).

We now define our associated BP message passing algorithm: LABP (Loopy
Annealing BP).
We introduce the set $\orE$ of directed edges of $G$ comprising two
directed edges $u\to v$ and $v\to u$ for each undirected edge $uv \in
E$. For $\ore\in \orE$, we denote by $-\ore$ the edge with opposite
direction. 
With a slight abuse of notation, we denote by $\d v$ the set of
incident edges to $v\in V$ directed towards $v$.
The update rules of LABP depend on a parameter $z>0$ and are defined by
$m^{0}_{\ore}=0$ and for $t\geq 0$ and all $u,v$ neighbors in $G$:
\begin{eqnarray}
\label{eq:rec}m^{t+1}_{u\to v}(z) = \frac{z}{1+\sum_{w\in \d u\bac v}m^{t}_{w\to u}(z)},
\end{eqnarray}
where $\d u\bac v$ is the set of neighbors of $u$ in $G$ from which we
removed $v$ and with the convention that the sum over the empty set
equals zero.
We denote by $z\Rcal_G$ the mapping sending $\bold{m}^t(z)\in [0,\infty)^{\orE}$ to
$\bold{m}^{t+1}(z)=z\Rcal_G(\bold{m}^t(z))$. We also denote by $z\Rcal_{\ore}$ the local update
rule (\ref{eq:rec}): $m^{t+1}_{\ore}(z)=z\Rcal_{\ore}(\bold{m}(z))$.

\begin{theorem}\label{th1}
For any finite graph $G$ and $z>0$, LABP converges: $\lim_{t\to \infty}
m^t_{\ore}(z)= Y_{\ore}(z)$. 
For $z>0$, let $\bx(z)\in \R^E$ be defined by
\BEA
\label{def:xe}x_e(z) =
\frac{Y_{\ore}(z)Y_{-\ore}(z)}{z+Y_{\ore}(z)Y_{-\ore}(z)}\in (0,1).
\EEA
For $z>e^{|E|}$, we have $\bx(z)=\bx^*$ and $\bx^*\in \R^E$
is a maximum fractional matching of $G$. In particular, we have
\begin{eqnarray*}
\sum_{e\in E}x^*_e
= \nu^*(G) =\tau^*(G).
\end{eqnarray*}
\end{theorem}

\begin{example}
Consider the cycle with 3 nodes. Then, $Y(z)$ is the same for all edges and has to satisfy $Y(z) = z(1+Y(z))^{-1}$ so that we get $Y(z) = \sqrt{1/4+z}-1/2$. Finally, we find that the expression above equals $\nu^*(G) = 3/2$.
\end{example}

We now define a much simpler message passing algorithm and a simpler
expression for $\nu^*(G)$.
Given a set of $\{0,1\}$-valued messages $\bI$, we define a new set of
$\{0,1\}$-valued messages by:
\BEA
\label{def:Pcal}J_{u\to v} = \ind\left(\sum_{\ell\in \d u \bac v}I_{w\to u}=0\right),
\EEA
with the convention that the sum over the empty set equals zero.
We denote by $\Pcal_G$ the mapping sending $\bI\in \{0,1\}^{\orE}$ to
$\bJ=\Pcal_G(\bI)$ and as above, $\Pcal_{\ore}$ denotes the local
update rule. Note that $\Pcal_G$ corresponds to the max-product
algortihm presented in \cite{smw11} with all weights equal to one.
We define for each $v\in V$ and $\bI\in \{0,1\}^{\orE}$,
\BEA
\label{def:Fv}F_v(\bI) = 1\wedge\left(\sum_{u\in \d v}I_{u\to v}\right)+\left(1- \sum_{u\in \d v} I_{v\to u}\right)^+,
\EEA
where $a\wedge b = \min(a,b)$ and $(a)^+=\max(a,0)$.
The second part of the following theorem corresponds to Proposition 3.5 in \cite{soda12} applied in our setting.
\begin{theorem}\label{th2}
For any graph $G$, if $\bI=\Pcal_G\circ\Pcal_G(\bI)$ then the vector $(\frac{F_v(\bI)}{2},\:v\in V)$ is a fractional vertex cover of $G$. 
Moreover, we have
\BEA
\label{eq:nuI}\nu^*(G) = \inf_{\bI} \sum_{v\in V}\frac{F_v(\bI)}{2},
\EEA
where the infimum is over the solutions of
$\bI=\Pcal_G\circ\Pcal_G(\bI)$.
\end{theorem}
By \cite{smw11}, if the LP relaxation (\ref{def:nu*}) has a unique optimum and
this optimum is integral, then iterating
the map $\Pcal_G$ will allow us to find the unique solution to the fixed point
equation $\bI=\Pcal_G(\bI)$. Indeed in this case, \cite{smw11} shows
that the following rule allows us to find the maximum matching from the messages $\bI$: put edge $e$ in the matching if and only if $I_{\ore}+I_{-\ore}=2$. Note that we
can then derive a minimum vertex cover from a maximum matching in linear time (Theorem
16.6 in \cite{sch03}).

We now show that LABP allows us to find $\bI$ achieving the
minimum in (\ref{eq:nuI}) and a minimum fractional vertex cover without any
restriction on $G$. 
\begin{proposition}\label{prop:IY}
Let $\bI^Y$ be the $\{0,1\}$-valued messages defined by $I_{\ore}^Y
=1$ if and only if $\lim_zY_{\ore}(z)=\infty$. Then $(F_v(\bI^Y)/2,\:v\in
V)$ is a minimum half-integral vertex cover, i.e. $2\nu^*(G)=\sum_v
F_v(\bI^Y)$.
In particular, $(F_v(\bI^Y),\:v\in V)$ is a 2-approximate solution to
vertex cover on $G$.
\end{proposition}
Recall that if the unique games conjecture is true, then vertex cover
cannot be approximated within any constant factor better than 2 as
shown by \cite{khot2008vertex}.

\begin{example}
Consider the cycle with 3 nodes. Then $I^Y_{\ore}=1$ for all oriented edges
and $F_v(\bI^Y)=1$ for all $v\in V$. Note that $\bI^Y$ is not the only
fixed point to $\bI=\Pcal_G\circ\Pcal_G(\bI)$, the all zeros vector is also
a solution. However the map $\Pcal_G$ has no fixed point and max-product BP
as defined in \cite{smw11} does not converge.
\end{example}

\subsection{Bipartite graphs}\label{sec:bi}
We now specialize our results to bipartite graphs.
If the graph is bipartite then the fractional matching polytope is
indeed the matching polytope, i.e. the convex hull of the incidence
vectors of matchings (Corollary 18.1(b) in \cite{sch03}) so that
$\nu^*(G)=\nu(G)$. By K\"onig's matching theorem (Theorem 16.2 in
\cite{sch03}), we also have in this case $\nu(G)=\tau(G)$. To summarize, a
direct application of Theorem \ref{th1} gives:
\begin{corollary}
If $G$ is bipartite, LABP computes the matching number
which is equal to the vertex cover number.
\end{corollary}

We now show that for any bipartite graph $G=(V=U\cup W,E)$, LABP allows us to define a minimum vertex cover.
For any $\bI\in\{0,1\}^{\orE}$, we consider the following subset $V(\bI)$
of vertices defined differently for vertices in $U$ and $W$ as follows:
\BEA
\label{def:VI1}\mbox{for } u\in U,\: u\in V(\bI) &\Leftrightarrow& \sum_{v\in \d u} I_{v\to
u} \geq 1,\\
\label{def:VI2}\mbox{for } w\in W,\: w\in V(\bI) &\Leftrightarrow&
\sum_{v\in \d w} \Pcal_{v\to w}(\bI) \geq 2.
\EEA

\begin{proposition}\label{prop:bi}
For any bipartite graph, the subset of vertices $V(\bI^Y)$ is a
minimum vertex cover, where $\bI^Y$ was defined in Proposition
\ref{prop:IY}.
\end{proposition}
\begin{example}
Consider the cycle with $4$ nodes. Again, we have $Y(z) = \sqrt{1/4+z}-1/2$
and the all-one vector is a fixed point of $\Pcal_G\circ \Pcal_G$. We
see that if we apply the results of the previous section, we have $F_v(\bI^Y)
=1$ for all $v\in V$ and we obtain a minimum fractional vertex cover. 
The above procedure (\ref{def:VI1}) and (\ref{def:VI2}) gives instead a
minimum vertex cover. 
Note also that $\Pcal_G$ has no fixed point so that the max-product BP of
\cite{smw11} does not converge.
\end{example}

\subsection{Results at positive temperature}\label{sec:pos}

In this section, we consider general graphs and LABP for
finite $z$.
We introduce the family of probability distributions on the set of
matchings parametrised by a parameter $z>0$:
\BEA
\label{eq:gibbs}\mu^z_G(\bB) = \frac{z^{\sum_e B_e}}{P_G(z)},
\EEA
where $P_G(z) =\sum_{\bB}z^{\sum_e B_e}\prod_{v\in
  V}\ind\left(\sum_{e\in \d v}B_e\leq 1\right)$.
For any finite graph, when $z$ tends to infinity, the distribution
$\mu_G^z$ converges to the uniform distribution over maximum matchings
so that we have
\BEA
\label{eq:MG}
\nu(G) =\lim_{z\to \infty} \sum_{e\in E} \mu^z_G(B_e=1).
\EEA
In statistical physics, this model is known as the monomer-dimer model
and its analysis goes back to the work of Heilmann and Lieb
\cite{heilmannlieb}, see also \cite{ss13} for a recent contribution in
theoretical computer science. The recursion (\ref{eq:rec}) has a natural
interpretation in term of the probability distribution
(\ref{eq:gibbs}) when the graph $G$ is a tree (see point (iii) in
Proposition \ref{prop:ft}).
The fact that this recursion is still useful for arbitrary graphs and
moreover allows us to study not only matching but vertex cover is
highly surprising.

In the rest of this section, we introduce the Bethe approximation
which is a standard approach to approximate the probability distribution
(\ref{eq:gibbs}). This approach will give results only for the
matching problem. The proofs of our results for the vertex cover
problem do not rely on this approximation and requires original
techniques based on subtle monotonicity arguments for the recursion
(\ref{eq:rec}) wich are presented in the Appendix \ref{sec:proofs}.

We define the internal energy $U(z)$ and the canonical entropy $S(z)$ as:
\begin{eqnarray*}
U_G(z) &=& -\sum_{e\in E} \mu_G^z(B_e=1),\\
S_G(z) &=& -\sum_{\bB} \mu^z_G(\bB)\ln \mu^z_G(\bB).
\end{eqnarray*}
The free entropy $\Phi_G(z)$ is then defined by
\begin{eqnarray*}
\Phi_G(z) = -U_G(z)\ln z  + S_G(z).
\end{eqnarray*}
A more conventional notation in the statistical physics literature
corresponds to an inverse temperature $\beta = \ln z$. A simple
computation shows that:
\begin{eqnarray*}
\Phi_G(z) = \ln P_{G}(z).
\end{eqnarray*}

Let $D(G)$ be the set of distribution over matchings, i.e. $\mu \in
D(G)$ if and only if $\mu(\bB \mbox{ is a matching in } G) =1$.
Let $\mu_G\in D(G)$.
For any $e\in E$, we define $\mu_{[G,e]}$ the marginal of $\mu_G$
restricted to $e$, i.e.
\begin{eqnarray*}
\mu_{[G,e]}(1)= 1-\mu_{[G,e]}(0)=\mu_G(B_e=1)= \sum_{\bB, B_e=1}
\mu_G(\bB).
\end{eqnarray*}
Similarly for any $v\in V$, we define $\mu_{[G,\d v]}$ the marginal of $\mu_G$
restricted to $\d v \subset E$.
For any $\mu_G\in D(G)$, the Bethe internal energy $U^B[\mu_G]$ and the Bethe entropy $S^B[\mu_G]$ are then defined by
\begin{eqnarray*}
U^B[\mu_G] &=& -\sum_{e\in E}\mu_{[G,e]}(1) \\
S^B[\mu_G] &=& -\sum_{v\in V}\sum_{\bb_{\d v}\in \{0,1\}^{|\d v|}}\mu_{[G,\d v]}(\bb_{\d
  v})\ln \left( \mu_{[G,\d v]}(\bb_{\d v})\right) \\
&&+ \sum_{e\in E}
\sum_{b_e\in \{0,1\}} \mu_{[G, e]}(b_e)\ln \left( \mu_{[G, e]}(b_e)\right)
\end{eqnarray*}
The Bethe free entropy $\Phi^B[\mu_G;z]$ is then defined by
\begin{eqnarray*}
\Phi^B[\mu_G;z] = -U^B[\mu_G]\ln z  + S^B[\mu_G]
\end{eqnarray*}

It is well known that if $G$ is a tree, i.e. acyclic graph, then we have
$\Phi^B[\mu_G^z;z] = \Phi_G(z)$ (see \cite{wj08}).

We first reformulate the Bethe free entropy function.
\begin{proposition}\label{prop:param}
Let $\mu_G\in D(G)$ be a distribution over matchings.
Define $\bx\in \R^E$ by
$x_e =\mu_{[G,e]}(1)$. Then we have $\bx\in FM(G)$ defined by (\ref{def:LG}) and 
\begin{eqnarray*}
U^B[\mu_G] &=& -\sum_{e\in E} x_e\\
S^B[\mu_G] &=& \frac{1}{2}\sum_{v\in V}\left\{\sum_{e\in \d v} -x_e\ln
x_e +(1-x_e)\ln (1-x_e) \right.\\
&&\left.- 2\left( 1-\sum_{e\in \d v} x_e\right)\ln
\left( 1-\sum_{e\in \d v}x_e\right)\right\},
\end{eqnarray*}
with the standard convention $0\ln 0 =0$.
\end{proposition}

We then have
\begin{proposition}\label{prop:conc}
The function $S^B(\bx)$ defined by
\begin{eqnarray*}
S^B(\bx)&=&\frac{1}{2}\sum_{v\in V}\left\{\sum_{e\in \d v} -x_e\ln
x_e +(1-x_e)\ln (1-x_e) \right.\\
&&\left.- 2\left( 1-\sum_{e\in \d v} x_e\right)\ln
\left( 1-\sum_{e\in \d v}x_e\right)\right\}
\end{eqnarray*}
is non-negative and concave on $FM(G)$ defined by (\ref{def:LG}).
\end{proposition}

We also define $U^B(\bx)=-\sum_{e\in E} x_e$ and  $\Phi^B(\bx;z)=
-U^B(\bx)\ln z +S^B(\bx)$.
Note that for any $\mu_G\in D(G)$, we have,
\BEAS
\Phi^B(\mu_G;z) = \Phi^B(\bx,z),
\EEAS
for $\bx$ defined by $x_e =\mu_{[G,e]}(1)$.
\begin{proposition}\label{prop:max}
Recall that $\bx(z)\in \R^E$ is defined by (\ref{def:xe}).
Then we have:
\begin{eqnarray*}
\sup_{\bx\in FM(G)}\Phi^B(\bx;z)&=&\Phi^B(\bx(z);z).
\end{eqnarray*}
\end{proposition}
In words, our BP algorithm is shown to maximize the Bethe free entropy
which in our case is concave. Note that in general, a BP fixed point
corresponds to a staionnary point of the Bethe free entropy, see \cite{yfw05}.

We now give a reparametrization of the Gibbs distribution.
For any vector $\bB\in\{0,1\}^{\orE}$, we denote by $\bB_{\d
  v}\in\{0,1\}^{\d v}$ its restriction to components in $\d v$. 
We first define the marginal probabilities
\BEAS
\mu_{\d v}(\bB_{\d v}) = \left( 1-\sum_{e\in \d v}x_e(z)\right)^{1-\sum_{e\in \d v} B_e}\prod_{e\in \d v} x_e(z)^{B_e},
\EEAS
and
\BEAS
\mu_e(B_e) = x_e(z)^{B_e}(1-x_e(z))^{1-B_e},
\EEAS
where $x_e(z)$ is defined by (\ref{def:xe}).
Given a graph $G=(V,E)$ and some set $F\subset E$, we define $d_F(v)$
as the degree of node $v$ in the subgraph induced by $F$. A
generalized loop is any subset $F$ such that $d_F(v) \neq 1$ for all
$v\in V$. We define $V(F)$ as the number of vertices covered by $F$,
i.e. vertices with $d_F(v)\geq 1$.
\begin{theorem}\label{th3}
For any graph $G$, we have for $z>0$,
\BEA
\label{eq:muG}\mu^z_G(\bB)=\frac{1}{Z} \frac{\prod_{v\in V} \mu_{\d v}(\bB_{\d
    v})}{\prod_{e\in E} \mu_e(B_e)},
\EEA
with
\BEA
\label{eq:Z}Z = 1+\sum_{\emptyset\neq F\subset E}(-1)^{V(F)}\prod_{v\in
  V}\left(d_F(v)-1\right)\prod_{e\in F}\frac{x_e(z)}{1-x_e(z)},
\EEA
where only generalized loops $F$ lead to non-zero terms in the sum of
(\ref{eq:Z}). Moreover, we have
\BEAS
\ln Z = \Phi_G(z)-\Phi^B(\bx(z);z).
\EEAS
\end{theorem}

Note that if $G$ is a tree, we recover that $Z=1$ and that our BP
algorithm computes exactly the marginals of the Gibbs distribution defined
by (\ref{eq:gibbs}).
However for general graphs, BP algorithm is not exact and equation
(\ref{eq:Z}) gives the exact correction term as a loop serie expansion
\cite{cc06}.
Explicit computation of these loops is in general intractable. Indeed
counting the total number of matchings $\exp\left(\Phi_G(1)\right)$ falls into the class of $\#
P$-complete problem. However equation (\ref{eq:Z}) can be used to
approximate such quantities by accounting for a small set of
significant loop corrections.

\section{Conclusion}\label{sec:conc}
We introduced an annealing BP algorithm for the vertex cover and matching
problems and showed its convergence, its relation to LP relaxation and
conditions for correctness. In contrast to previous results of this kind,
we do not rely on the a priori uniqueness of the solution to the
optimization problem. In view of the recent results \cite{salez12} and
\cite{llm13}, our approach should extend to more complex settings:
b-matching, capacited matching. Another direction worth investigating is
the question of the convergence time of our algorithm that we left open
(techniques used in \cite{sgg11} seems relevant).



\bibliographystyle{abbrv}


\section{Appendix: Proofs}\label{sec:proofs}
\subsection{Convergence of BP}
Given a set of messages
$\bX$, we define a new set of messages $\bY$ by:
\BEA
\label{eq:recR}Y_{u\to v} = \frac{1}{1+\sum_{w\in \d u\bac v}X_{w\to u}},
\EEA
with the convention that the sum over the empty set equals zero.
We denote by $\Rcal_G$ the mapping sending $\bX\in [0,\infty)^{\orE}$ to
$\bY=\Rcal_G(\bX)$. We also denote by $\Rcal_{\ore}$ the local update
rule (\ref{eq:recR}): $Y_{\ore}=\Rcal_{\ore}(\bX)$.
Note that the mapping $z\Rcal_G$ defined in (\ref{eq:rec}) is simply the
mapping multiplying by $z$ each component of the output of the mapping
$\Rcal_G$ (making the notation consistent).
\begin{proposition}\label{prop:ft}
\begin{itemize}
\item[(i)] For any finite graph $G$ and $z>0$, the fixed point equation:
\BEA
\label{eq:fp}\bX = z\Rcal_G(\bX)
\EEA
has a unique attractive solution denoted $\bY(z)\in (0,+\infty)^{\orE}$.
\item[(ii)] The function $z\mapsto \bY(z)$ is non-decreasing and the
  function $z\mapsto \frac{\bY(z)}{z}$ is non-increasing for $z>0$.
\item[(iii)] If in addition, $G$ is a finite tree, then for all $e\in E$, the law of $B_e$ under
$\mu^z_G$ is a Bernoulli distribution with
\BEA
\label{eq:margin}\mu^z_G\left(B_e=1\right) =
\frac{Y_{\ore}(z)\Rcal_{-\ore}(\bY(z))}{1+Y_{\ore}(z)\Rcal_{-\ore}(\bY(z))}.
\EEA
\end{itemize}
\end{proposition}
Comparisons between vectors are always componentwise.
Note that the right-hand side of (\ref{eq:margin}) does not depend on
the choice of orientation of the edge $e$ as $\bY(z)$ satisfies (\ref{eq:fp}).
Before proving this proposition, we define for all $v\in V$, the
following function of the messages $(Y_{\ore},\:\ore\in \d v)$,
\BEA
\label{def:D1}\Dcal_v(\bY) &=& \sum_{\ore\in \d v} \frac{Y_{\ore}\Rcal_{-\ore}(\bY)}{1+Y_{\ore}\Rcal_{-\ore}(\bY)}\\
\label{def:D2}&=& \frac{\sum_{\ore \in \d v} Y_{\ore}}{1+\sum_{\ore \in \d v} Y_{\ore}}.
\EEA
In view of point (iii) of Proposition \ref{prop:ft}, we see that if
the graph $G$ is a tree, $\Dcal_v(\bY(z))$ is simply the probability for vertex $v$ to be
covered by a matching distributed according to $\mu_G^z$.
In particular, when $G$ is a tree, we can rewrite (\ref{eq:MG}) as
\BEA
\label{eq:MGb} \nu(G) =\lim_{z\to \infty} \frac{1}{2}\sum_{v\in V}\Dcal_v(\bY(z)).
\EEA

\begin{proof}
For the first point, we follow the proof of Theorem 3 in
\cite{salez12}. Let $z>0$ and define the sequence of messages: $\bX^0(z)=0$ and
for $t\geq 0$,
\BEA
\label{eq:Xt}X^{t+1}_{u\to v}(z) &=& \frac{z}{1+\sum_{w \in \d u\bac
    v}X^t_{w\to u}(z)}.
\EEA
The sequence $\bX^{2t}(z)$ (resp. $\bX^{2t+1}(z)$) is non-decreasing
(resp. non-increasing). We define $\lim_{t\to \infty}\uparrow
\bX^{2t}(z)=\bX^{-}(z)$ and $\lim_{t\to
  \infty}\downarrow\bX^{2t+1}(z)=\bX^{+}(z)$. For any $\bY(z)$ fixed point of
(\ref{eq:fp}), a simple induction shows that
\BEAS
0\leq\bX^{2t}(z)\leq \bX^{-}(z)\leq \bY(z)\leq \bX^{+}(z)\leq \bX^{2t+1}(z)\leq z.
\EEAS
We now prove that $\bX^{-}(z) =  \bX^{+}(z)$ finishing the proof of
the first point.
Note that we have
$\bX^{+}(z)=z\Rcal_G(\bX^{-}(z))$ and
$\bX^{-}(z)=z\Rcal_G(\bX^{+}(z))$. In particular for any $z>0$, we
have
$X^{+}_{\ore}(z)\Rcal_{-\ore}(\bX^+(z))=X^{-}_{-\ore}(z)\Rcal_{\ore}(\bX^-(z))$
so that in view of (\ref{def:D1}), we have
\BEA
\label{eq:fmtp}\sum_{v\in V}\Dcal_v(\bX^+(z)) = \sum_{v\in V}\Dcal_v(\bX^-(z)).
\EEA
We see from (\ref{def:D2}) that for each $v\in V$, $\Dcal_v$ is
an increasing function of the $(X_{\ore},\:\ore\in \d v)$, so that
(\ref{eq:fmtp}) together with $\bX^-(z)\leq \bX^+(z)$ imply the desired
result.

We now prove that $z\mapsto \frac{\bX^t(z)}{z}$ and
$z\mapsto \bX^t(z)$ are respectively non-increasing and
non-decreasing, this implies point (ii).
We prove it by induction on $t$: consider $z\leq z'$ if $\bX^t(z)\leq
\bX^t(z')$ then by (\ref{eq:Xt}) we have $\frac{\bX^{t+1}(z)}{z}\geq
\frac{\bX^{t+1}(z')}{z'}$ and if $\frac{\bX^t(z)}{z}\geq
\frac{\bX^t(z')}{z'}$ then again by (\ref{eq:Xt}), we have $\bX^{t+1}(z)\leq \bX^{t+1}(z')$.

We consider now the case where $G$ is a tree. For any directed edge
$u\to v$, we define $T_{u\to v}$ as the subtree containing $u$ and
$v$ and obtained from $G$ by removing all incident edges to $v$ except
the edge $uv$. A simple computation shows that
\BEAS
\frac{\mu_{T_{u\to v}}^z(B_{uv=1})}{\mu_{T_{u\to
      v}}^z(B_{uv=0})}=\frac{z}{1+\sum_{w\in \d u\bac v}
  \frac{\mu_{T_{w\to u}}^z(B_{wu=1})}{\mu_{T_{w\to u}}^z(B_{wu=0})}}.
\EEAS
This directly implies that for a finite tree, $Y_{u\to v}(z) =
\frac{\mu_{T_{u\to v}^z(B_{uv=1})}}{\mu_{T_{u\to v}^z(B_{uv=0})}}$.
Then a simple computation shows that
\BEAS
\frac{\mu_{G}^z(B_{uv=1})}{\mu_{G}^z(B_{uv=0})}
&=& \frac{Y_{u\to v}(z) Y_{v\to u}(z)}{z}\\
&=& Y_{u\to v}(z)\Rcal_{v\to u}(\bY(z)),
\EEAS
which directly implies (\ref{eq:margin}).
\end{proof}

\subsection{Zero temperature limit}
In order to compute the matching number, we must let $z$ tend to
infinity in $\bY(z)=z\Rcal_G(\bY(z))$. Iterating once this recursion,
we get $\bY(z) = z\Rcal_G(z\Rcal_G(\bY(z)))$. Note that we have for
any $z>0$,
\BEAS
z\Rcal_{u\to v}(z X) = \frac{1}{z^{-1} + \sum_{w\in \d u\bac v}
  X_{w\to u}}
\EEAS
Hence we can define for any $\bX\in (0,1]^{\orE}$, $\Qcal_G(\bX)
=\lim_{z\to \infty}\uparrow z\Rcal_G(z \bX)\in (0,\infty]^{\orE}$ by its local update rule:
\BEA
\label{def:Q} \Qcal_{u\to v} (\bX) =\frac{1}{\sum_{w\in \d u\bac v}
  X_{w\to u}},
\EEA
with the conventions $1/0=\infty$ and the sum over the empty set equals
zero (in particular, if $u$ is a leaf of the graph $G$, then
$\Qcal_{u\to v}(\bX) = \infty$).

By point (ii) of Proposition \ref{prop:ft}, we can define $\lim_{z\to
  \infty}\uparrow \bY(z)=\bY\in [0,\infty]^{\orE}$ and $\lim_{z\to
  \infty}\downarrow\frac{\bY(z)}{z}=\bX\in [0,1]^{\orE}$. Then, we have
\BEA
\label{eq:fpzero}\bX = \Rcal_G(\bY) \mbox{ and, } \bY = \Qcal_G(\bX), 
\EEA
provided we can extend the maps $\Rcal_G$ and $\Qcal_G$ continuously
from their respective domains $[0,\infty)^{\orE}$ and $(0,1]^{\orE}$
to their compactifications $[0,\infty]^{\orE}$ and $[0,1]^{\orE}$
respectively. This can be done easily as follows: if there exists
$w\in \d u\bac v$ with $Y_{w\to u}=\infty$, then we set $\Rcal_{u\to
  v}(\bY)=0$; and if $X_{w\to u}=0$ for all $w\in \d u\bac v$, then we
set $\Qcal_{u\to v}(\bX) = \infty$.

\begin{lemma}
\label{lem:min}Let $\lim_{z\to \infty}\uparrow \bY(z)=\bY\in [0,\infty]^{\orE}$. Then
$\bY$ is the
smallest solution to the fixed point equation $\bY=\Qcal_G\circ
\Rcal_G(\bY)$.
\end{lemma}
\begin{proof}
Let $\bZ=\Qcal_G\circ\Rcal_G(\bZ)$. For any $z>0$, we have for any
$\bX\in [0,1]^{\orE}$, $z\Rcal_G(z\bX)\leq \Qcal_G(\bX)$ so that an
easy induction implies that $\bX^{2t}(z)\leq \bZ$ where $\bX^{2t}(z)$ is
the sequence defined in the proof of Proposition \ref{prop:ft}.
Letting first $t$ and then $z$ tend to infinity, allows us to conclude.
\end{proof}

Note that thanks to (\ref{def:D2}), we can extend the functions
$\Dcal_v(\bY)$ continuously on $[0,\infty]^{\orE}$ by setting
$\Dcal_v(\bY)= 1$ as soon as there exists $Y_{\ore}=\infty$ for
$\ore\in \d v$.
To summarize, we have for each $v\in V$,
\BEA
\label{eq:Dv0}\lim_{z\to \infty}\Dcal_v(\bY(z)) = \Dcal_v(\bY)\leq 1,
\EEA
where $\bY$ is the smallest solution to the fixed point equation
$\bY=\Qcal_G\circ\Rcal_G(\bY)$ that can be written as:
\BEA
\label{eq:rec0}Y_{u\to v} = \frac{1}{\sum_{w\in \d u\bac
    v}\frac{1}{1+\sum_{w'\in \d w\bac u}Y_{w'\to w}}},
\EEA
with the conventions $1/0=\infty$ and $1/\infty =0$ and the sum over
the empty set equals zero.

\begin{lemma}\label{lem:Dv}
We have for any $\bY\in [0,\infty]^{\orE}$ and $v\in V$,
\BEA
\nonumber \Dcal_v(\bY) &=& \sum_{\ore\in \d
  v}\frac{Y_{\ore}\Rcal_{-\ore}(\bY)}{1+Y_{\ore}\Rcal_{-\ore}(\bY)}\ind\left(
  Y_{\ore}<\infty\right)\\
\label{def:D0}&&+ \ind\left( \exists \ore\in \d v,
  Y_{\ore}=\infty\right),
\EEA
where the first sum on the right-hand side should be understood as a
sum over $\ore\in \d v$ with $Y_{\ore}<\infty$.
\end{lemma}
Note that since $\Rcal_{-\ore}(\bY)\in [0,1]$, the product
$Y_{\ore}\Rcal_{-\ore}(\bY)$ is always well-defined in the expression above.
\begin{proof}
We only need to consider the case where there exists $\ore\in \d v$
such that $Y_{\ore}=\infty$. By the discussion before the lemma, we have in this
case $\Dcal_v(\bY)=1$.
Hence we need to prove that the first term on the right-hand side of
(\ref{def:D0}) vanishes. This follows form the following fact: let
$\ore'\in \d v\bac \ore$, then $Y_{\ore}=\infty$ implies that
$\Rcal_{-\ore'}(\bY)=0$.
\end{proof}

For the messages $\bY\in [0,\infty]^{\orE}$ (resp. $\bX\in
[0,1]^{\orE}$) defined in (\ref{eq:fpzero}), we define the
$\{0,1\}$-valued messages $\bI^Y$ (resp. $\bI^X$) by $I_{u\to v}^Y
=\ind(Y_{u\to v}=\infty)$ (resp. $I^X_{u\to v} = \ind(X_{u\to v}>0)$.
It follows directly from (\ref{eq:fpzero}) and the definition of $\Pcal_G$ (\ref{def:Pcal}) that
\BEA
\label{eq:fp0}\bI^Y = \Pcal_G(\bI^X),\mbox{ and, } \bI^X=\Pcal_G(\bI^Y).
\EEA

We now show that for any finite graph $G$, the right-hand term in
(\ref{eq:MGb}) is a function of $\bI^X$ and $\bI^Y$ only.

For any $Y\in [0,\infty]$, we define $I(Y)=\ind(Y=\infty)$ and
still denote by $I$ the function acting similarly on vectors
componentwise, i.e. if $\bI=I(\bY)$ then $I_{\ore}=I(Y_{\ore})$.
\begin{lemma}\label{lem:mon}
For $\bY\in [0,\infty]^{\orE}$, we define $\bY'=\Qcal_G\circ
\Rcal_G(\bY)$. If $\bY\geq $ (resp. $\leq$) $\bY'$, then 
\BEAS
\sum_{v\in V}\Dcal_v(\bY) \geq \mbox{ (resp. $\leq$) } \sum_{v\in V} F_v(I(\bY')),
\EEAS
where $F_v$ was defined in (\ref{def:Fv}).
\end{lemma}
\begin{proof}
Suppose $\bY'\leq \bY$, then using Lemma \ref{lem:Dv}, we get
\BEAS
\sum_v \Dcal_v(\bY) &\geq & \underbrace{\sum_{\ore\in \orE}\frac{Y'_{\ore}\Rcal_{-\ore}(\bY)}{1+Y'_{\ore}\Rcal_{-\ore}(\bY)}\ind\left(
  Y'_{\ore}<\infty\right)}_{A}\\
&&+\sum_{v\in V}\ind\left( \exists \ore\in \d v,
  Y'_{\ore}=\infty\right).
\EEAS
For the first term $A$, denote $\bX=\Rcal_G(\bY)$ so that
$\bY'=\Qcal_G(\bX)$. 
Then we have
\BEAS
A &=& \sum_{\ore\in
  \orE}\frac{\Qcal_{\ore}(\bX)X_{-\ore}}{1+\Qcal_{\ore}(\bX)X_{-\ore}}\ind\left(\Qcal_{\ore}(\bX)<\infty
\right)\\
&=&\sum_{\ore\in
  \orE}\frac{\Qcal_{-\ore}(\bX)X_{\ore}}{1+\Qcal_{-\ore}(\bX)X_{\ore}}\ind\left(\Qcal_{-\ore}(\bX)<\infty
\right)\\
&=& \sum_{v\in V} \underbrace{\sum_{\ore\in \d v}\frac{X_{\ore}\Qcal_{-\ore}(\bX)}{1+X_{\ore}\Qcal_{-\ore}(\bX)}\ind\left(\Qcal_{-\ore}(\bX)<\infty
\right)}_{B_v}.
\EEAS
We now prove that
\BEA
\label{eq:Bv}B_v = \left(1- \sum_{\ore\in \d v} I_{-\ore}(\bY')\right)^+.
\EEA
First note that if $J(\bX)$ is defined by $J_{\ore}(\bX)=\ind(X_{\ore}>0)$, then we have $\Pcal_G(J(\bX))=I(\bY')$. Hence if $\sum_{\ore\in \d v} I_{-\ore}(\bY')=0$, then $\exists w\neq w'$ both in $\d v$ with $X_{w\to v}X_{w'\to v}>0$. This in turn implies that $0<\Qcal_{-\ore}(\bX)<\infty$ for all $\ore\in \d v$, so that in this case we have
\BEAS
B_v = \sum_{\ore\in \d v}\frac{X_{\ore}}{\Qcal_{-\ore}(\bX)^{-1}+X_{\ore}}=1.
\EEAS
 
Note now that if $B_v>0$, there must
exists $\ore\in \d v$ such that $X_{\ore}>0$ and
$\Qcal_{-\ore}(\bX)<\infty$ and this last constraint implies that
there exists $\ore'\neq \ore$ with $\ore' \in \d v$ with $X_{\ore'}>0$. In particular, we have $B_v=1$ and $\sum_{\ore\in \d v} I_{-\ore}(\bY')=0$. This finished the proof of (\ref{eq:Bv}).
The lemma then follows.
\end{proof}
We are now ready to state our first main result for finite graphs:
\begin{proposition}\label{prop:limD}
For any finite graph $G$, we have
\BEAS
\sum_{v\in V} \Dcal_v(\bY)=\lim_{z\to \infty}\sum_{v\in V}\Dcal_v(\bY(z)) = \inf_{\bI} \sum_{v\in V}F_v(\bI),
\EEAS
where the infimum is over the solutions of
$\bI=\Pcal_G\circ\Pcal_G(\bI)$.
\end{proposition}
\begin{proof}
Let $\bY=\lim_{z\to \infty}\uparrow \bY(z)$ and recall that we denoted
$\bI^Y =I(\bY)$ so that $\bI^Y =
\Pcal_G\circ\Pcal_G(\bI^Y)$ by (\ref{eq:fp0}).
By Lemma \ref{lem:mon} and (\ref{eq:Dv0}), we have
\BEA
\label{eq:limIY}\lim_{z\to \infty}\sum_{v\in
  V}\Dcal_v(\bY(z)) =\sum_{v\in V} \Dcal_v(\bY)=\sum_{v\in V} F_v(\bI^Y).
\EEA

We need to prove that if $\bI=\Pcal_G\circ \Pcal_G(\bI)$ then we
have $\sum_v F_v(\bI)\geq \sum_{v\in V} \Dcal_v(\bY)$.
For any such $\bI$, we define $\bW^0$ as follows:
\BEAS
W^0_{\ore} = \left\{
\begin{array}{ll}
\infty & \mbox{if }I_{\ore} =1\\
0&\mbox{otherwise.}
\end{array}
\right.
\EEAS
Then let $\bW^{k+1}=\Qcal_G\circ \Rcal_G(\bW^k)$ for $k\geq 0$. A simple
induction shows that $I(\bW^{k+1})=\Pcal_G\circ\Pcal_G(I(\bW^k))=\bI$
for all $k\geq 0$. In particular, $\bW^0\leq \bW^1$ and again by
induction, we see that the sequence $\{\bW^k\}_k$ is non-decreasing and we denote by
$\bW^{\bI}$ its limit. 
Applying Lemma \ref{lem:mon} to $\bW^k$, we get
\BEAS
\sum_{v\in V} \Dcal_v(\bW^{k}) \leq \sum_{v\in V} F_v(\bI).
\EEAS
Taking the limit $k\to \infty$, we obtain
\BEAS
\sum_{v\in V} \Dcal_v(\bW^{\bI}) \leq \sum_{v\in V} F_v(\bI)
\EEAS
Moreover $\bY$ being the smallest solution to the fixed point equation
$\bY=\Rcal_G\circ \Rcal_G(\bY)$, we have $\bY\leq \bW^{\bI}$ and using
the fact that $\Dcal_v$ is increasing, we get
\BEAS
\sum_{v\in V} \Dcal_v(\bY)\leq \sum_{v\in V} \Dcal_v(\bW^{\bI})\leq \sum_{v\in V} F_v(\bI),
\EEAS
which concludes the proof.
\end{proof}

We now prove
\begin{lemma}\label{lem:fraccov}
If $\bI=\Pcal_G\circ\Pcal_G(\bI)$, then $(F_v(\bI)/2,\:v \in V)$ is a half-integral vertex cover.
\end{lemma}
\begin{proof}
We need to prove that for all $(uv)\in E$, $F_u(\bI)+F_v(\bI)\geq 2$.
This follows easily from the fact that if $\sum_{w\in \d u}I_{w\to u}=0$ then we have $I_{v\to w}=0$ for all $w\in \d v$ and hence $\left(1-\sum_{w\in \d v}I_{v\to w}\right)^+=1$. Hence we have
\BEAS
F_u(\bI)+F_v(\bI) &\geq& \sum_{w\in \d u}I_{w\to u}+\sum_{w\in \d v}I_{w\to v}\\
&+&\ind\left(\sum_{w\in \d u}I_{w\to u}=0\right)+\ind\left(\sum_{w\in \d v}I_{w\to v}=0\right)\\
&\geq& 2.
\EEAS
\end{proof}

In the rest of this subsection, the
graph $G = (U\cup W,E)$ is assumed to be bipartite. 
The following lemma shows that $V(\bI^Y)$ is a vertex cover.
\begin{lemma}\label{lem:vc}
For any $\bI=\Pcal_G\circ \Pcal_G(\bI)$, the set $V(\bI)$ defined by (\ref{def:VI1}) and
(\ref{def:VI2}) is a vertex cover.
\end{lemma}
\begin{proof}
Consider $u\in U$, $w\in W$ and $(uw)\in E$. We denote $\bJ=\Pcal_G(\bI)$
so that $\bI=\Pcal_G(\bJ)$.
The fact that $u\notin V(\bI)$ implies that $J_{u\to w}=1$ and since
$I_{w\to u}=0$, there exists $v\in \d w\bac u$ such that $J_{v\to w}=1$ so
that $w\in V(\bI)$.
Similarly if $w\notin V(\bI)$ then $\sum_{v\in \d w}J_{v\to w}\leq
1$. Hence if $J_{u\to w}=1$, then $I_{w\to u}=1$ and if $J_{u\to w}=0$ then
there exists $v\in \d u\bac w$ with $I_{v\to u}=1$. So in both cases, $u\in
V(\bI)$.
\end{proof}
\begin{lemma}\label{lem:VI}
Let $\bI$ achieving $\inf_{\bI} \sum_{v\in V}{F_v(\bI)}$ where the infimum is over the solutions of
$\bI=\Pcal_G\circ \Pcal_G(\bI)$. Then, the size of the set $V(\bI)$ is
$\frac{1}{2}\sum_v F_v(\bI)$.
\end{lemma}
\begin{proof}
Again, we denote $\bJ=\Pcal_G(\bI)$
so that $\bI=\Pcal_G(\bJ)$. First note that $\left( 1-\sum_{\ore \in \d v} I_{-\ore}\right)^+ =
\ind\left( \sum_{\ore \in \d v}J_{\ore}\geq 2\right)$. Hence we have
$\sum_v F_v(\bI)=A+B$, with
\BEAS
A &=& \sum_{u\in U}\left(1\wedge\sum_{w\in \d u}I_{w\to u}\right) +\sum_{w\in W}
\ind\left(\sum_{u\in \d w} J_{u\to w} \geq 2\right),\\
B &=& \sum_{w\in W}\left(1\wedge\sum_{u\in \d w}I_{u\to w}\right) +\sum_{u\in U}
\ind\left(\sum_{w\in \d u} J_{w\to u} \geq 2\right).
\EEAS
Clearly $A$ is the size of the set $V(\bI)$, so we need only to show that
$A=B$.
Note that $A$ depends only on messages form $\bI$ from nodes in $W$ to nodes in $U$
and $B$ depends only on the remaining messages in $\bI$.
Assume that $A<B$ and consider
\BEAS
B' = \sum_{w\in W} \left(1\wedge \sum_{u\in \d w}J_{u\to w}
\right)+\sum_{u\in U} \ind\left( \sum_{w\in \d u I_{w\to u}}\right).
\EEAS
Note that we have
\BEAS
1\wedge \sum_{u\in \d w}J_{u\to w} &=& \left(1-\sum_{u\in \d w}I_{w\to
    u}\right)^+\\
&=& \ind\left( \sum_{u\in \d w} J_{u\to w}\geq 2\right).
\EEAS
In particular we have $B'\leq A<B$. Moreover if $\bK$ is such that for any
$w\in W$ and $u\in U$, $K_{w\to u}=I_{w\to u}$ and $K_{u\to w}=J_{u\to
  w}$, then we have $\sum_vF_v(\bK)=A+B'< \sum_vF_v(\bI)$ and
$\bK=\Pcal_G\circ \Pcal_G(\bK)$ contradicting the minimality of $\bI$.
\end{proof}

\subsection{Positive temperature}

We first prove Proposition \ref{prop:param}. For $\mu_G\in D(G)$, we have
$\mu_G(\sum_{e \in \d v} B_e\leq 1)=1$ so that by the linearity of
expectation, $\sum_{e \in \d v}\mu_{[G,e]}(1)\leq 1$ and the vector $\bx$
with component $x_e =\mu_{[G,e]}(1)$ is in $FM(G)$.
Now for each $v\in V$ and $e\in \d v$, we must have
\BEAS
\mu_{[G,\d v]}(b_e=1)&=& \sum_{b_f \in \{0,1\},\:f\in \d v\bac e}\mu_{[G,\d
  v]}(\bb_{\d v})\\ 
&=& \mu_{[G,\d v]}(b_e=1,\:b_f =0, f\in \d v\bac e)\\
&=& x_e.
\EEAS
It then follows that
\BEA
\label{def:mumarg}\mu_{[G,\d v]}(\bB_{\d v}) = \left( 1-\sum_{e\in \d
    v}x_e\right)^{1-\sum_{e\in \d v} B_e}\prod_{e\in \d v} x_e^{B_e},
\EEA
and the formula for $S^B[\mu_G]$ follows.

We now give a lemma implying the concavity of the Bethe entropy,
Proposition \ref{prop:conc}.
For $k\in \N$, we define $\Delta^k=\{\bx \in \R^k,\: x_i\geq 0,
\sum_{i=1}^k x_i\leq 1\}$.
\begin{lemma}
Let $g:\Delta^k \to \R$ be defined by
\begin{eqnarray*}
g(\bx)&=& -\sum_i x_i\ln x_i +\sum_i(1-x_i)\ln (1-x_i)\\
&&-
2\left(1-\sum_{i}x_i\right) \ln \left( 1-\sum_{i}x_i\right).
\end{eqnarray*}
For $k\geq 1$, $g$ is concave. Moreover, we have
\begin{eqnarray*}
\frac{\d g}{\d x_i} &=& \ln\left(\frac{\left(1-\sum_{j}x_j\right)^2}{x_i(1-x_i)}\right).
\end{eqnarray*}
\end{lemma}
\begin{proof}
From Theorem 20 in \cite{von13}, we know that the function  
\begin{eqnarray*}
h(\bx)&=& -\sum_i x_i\ln x_i +\sum_i(1-x_i)\ln (1-x_i)\\
&&-
\left(1-\sum_{i}x_i\right) \ln \left( 1-\sum_{i}x_i\right)\\
&&+\left(\sum_{i}x_i\right)\ln\left(\sum_{i}x_i\right)
\end{eqnarray*}
is non-negative and concave on $\Delta^k$.
We have
\begin{eqnarray*}
g(\bx)&=&h(\bx)+H\left(\sum_{i}x_i \right),
\end{eqnarray*}
where $H(p)=-p\ln p -(1-p)\ln (1-p)$ is the entropy of a Bernoulli
random variable and is concave in $p$.
\end{proof}

We now prove Proposition \ref{prop:max}.
For $e=(uv)\in E$ and $\bx\in \accentset{\circ}{\Delta^k}$ (the interior of
$\Delta^k$), we have
\begin{eqnarray*}
\frac{\d \Phi^B(\bx;z)}{\d x_e} &=& -\ln z\\
&&+\ln \left(\frac{\left(1-\sum_{f\in \d v}x_f\right)\left(1-\sum_{f\in \d u}x_f\right)}{x_e(1-x_e)} \right).
\end{eqnarray*}
Hence, we have $\frac{\d \Phi^B(\bx;z)}{\d x_e}=0$ if and only if
\begin{eqnarray}
\label{eq:x(1-x)}x_e(1-x_e) =z \left(1-\sum_{f\in \d v}x_f\right)\left(1-\sum_{f\in \d
    u}x_f\right).
\end{eqnarray}
Note that $\sum_{f\in \d v}x_f(z)=\Dcal_v(\bY(z))$, so that we have by (\ref{def:D2})
\begin{eqnarray*}
\left(1-\sum_{f\in \d v}x_f(z)\right)&=& \left(1-\frac{\sum_{\ore \in
      \d v} Y_{\ore}(z)}{1+\sum_{\ore \in \d v} Y_{\ore}(z)}\right)\\
&=& \left( 1+\sum_{\ore \in \d v} Y_{\ore}(z)\right)^{-1}
\end{eqnarray*}
We have for $e=(uv)\in E$, 
\BEAS
x_e(z) &=& \frac{Y_{u\to v}(z)}{\frac{z}{Y_{v\to u}(z)}+Y_{u\to v}(z)},
\EEAS
and using the fact that $\bY(z)z=z\Rcal_G(\bY(z)$, we get
\BEAS
x_e(z)=\frac{Y_{u\to v}(z)}{1+\sum_{w\in \d v}Y_{w\to v}(z)}= \frac{Y_{v\to u}(z)}{1+\sum_{w\in \d u}Y_{w\to u}(z)}.
\end{eqnarray*}
Hence, evaluating (\ref{eq:x(1-x)}) at $x_e(z)$, we get
\begin{eqnarray*}
x_e(z)(1-x_e(z))\frac{Y_{u\to v}(z)Y_{v\to u}(z)}{z} = x_e(z)^2,
\end{eqnarray*}
so that
\begin{eqnarray*}
\frac{x_e(z)}{1-x_e(z)} = \frac{Y_{u\to v}(z)Y_{v\to u}(z)}{z},
\end{eqnarray*}
which follows from the definition of $x_e(z)$ in
(\ref{def:xe}). Proposition \ref{prop:max} follows.

We also note that the following equality (which will be used later) is true
for $\bx(z)$ defined by (\ref{def:xe}):
\BEA
\label{eq:xe(z)} \frac{x_e(z)(1-x_e(z))}{z}=  \left(1-\sum_{e'\in \d u}x_{e'}(z)\right)\left(1-\sum_{e'\in \d v}x_{e'}(z)\right)
\EEA


\subsection{Proofs of the main results of Sections \ref{sec:fmvc} and \ref{sec:bi}}\label{sec:proofth12}

We now prove Theorems \ref{th1} and \ref{th2}. First by Proposition
\ref{prop:ft} (i), we have $\lim_{t\to \infty}m^t_{\ore}(z)
=Y_{\ore}(z)$. 
To end the proof of Theorem \ref{th1}, we need to show that $\bx(z)$
converges as $z\to\infty$ to a maximum fractional matching. First note
that $S^B(\bx)\leq |E|$ so that for $z$ sufficiently large,
Proposition \ref{prop:max} implies that $\bx(z)$ is on the optimal
face of $FM(G)$, i.e. $\sum_e x_e(z) = \nu^*(G)$ and maximizes the
function $S^B(\bx)$.

The first statement of Theorem \ref{th2} is exactly Lemma \ref{lem:fraccov}
and the second statement is Proposition \ref{prop:limD}. Proposition \ref{prop:IY}
follows from (\ref{eq:limIY}) and Lemma \ref{lem:fraccov}.
Proposition \ref{prop:bi} follows from Lemmas \ref{lem:vc}, \ref{lem:VI}
and the fact that $\bI^Y$ achieves the minimum in this last lemma (see Proposition
\ref{prop:IY}).


\subsection{Loop series expansion}

In this section, we prove Theorem \ref{th3}. The fact that $\mu^z_G$ can be
written as (\ref{eq:muG}) (called tree-based reparameterization in
\cite{wj08}) follows from a direct application of the definitions.

To simplify notation, we write in the proofs $x_e$ instead of $x_e(z)$.

\begin{lemma}\label{lem:fracmu}
For any $v\in V$, $z>0$, we have
\BEAS
\frac{\mu_{\d v}(\bB_{\d v})}{\prod_{e\in \d v} \mu_e(B_e)} =1-\sum_{S\subset \d v}(-1)^{|S|}\left(|S|-1\right) \prod_{e\in S}\frac{B_e-x_e(z)}{1-x_e(z)}.
\EEAS
\end{lemma}
\begin{proof}
Note that if $B_f=1$, the left-hand side is equal to $\prod_{e\neq f} (1-x_e)^{-1}$, while if $\sum_{e\in \d v}B_e=0$, it is equal to $\frac{1-\sum_{e\in \d v}x_e}{\prod_{e\in \d v}(1-x_e)}$.
We need to check that the right-hand side agrees in these two cases.
Let consider the case $B_f=1$, then the right-hand side (denoted $R$) equals:
\BEAS
R &=& 1-\sum_{|S|\geq 1,f\notin S}(-1)^{|S|}\left(|S|-1\right)\prod_{e\in S}\frac{-x_e}{1-x_e}\\
&&- \sum_{|S|\geq 1,f\in S}(-1)^{|S|}\left(|S|-1\right)\prod_{e\in S,e\neq f}\frac{-x_e}{1-x_e}\\
&=& 1-\sum_{|S|\geq 1,f\notin S}(-1)^{|S|+1}\prod_{e\in S}\frac{-x_e}{1-x_e}\\
&=& 1+\sum_{|S|\geq 1,f\notin S}\frac{\prod_{e\in S}x_e\prod_{e'\notin S, e'\neq f}(1-x_{e'})}{\prod_{e\neq f}(1-x_e)}\\
&=&\frac{1}{\prod_{e\neq f}(1-x_e)}.
\EEAS
A similar computation shows the second case.
\end{proof}

The following lemma shows (\ref{eq:Z}).
\begin{lemma}
We have
\BEAS
Z = 1-\sum_{\emptyset\neq F\subset E}(-1)^{V(F)}\prod_{v\in
  V}\left(d_F(v)-1\right)\prod_{e\in F}\frac{x_e(z)}{1-x_e(z)}.
\EEAS
\end{lemma}
\begin{proof}
By definition, we have
\BEAS
Z = \sum_{\bB} \prod_e \mu_e(B_e)\prod_v\frac{\mu_{\d v}(\bB_{\d v})}{\prod_{e\in \d v} \mu_e(B_e)}.
\EEAS
By Lemma \ref{lem:fracmu}, we have
\BEAS
P&:=&\prod_v\frac{\mu_{\d v}(\bB_{\d v})}{\prod_{e\in \d v} \mu_e(B_e)}\\
&=&\prod_v\left(1+\sum_{S\subset \d v}(-1)^{|S|-1}\left(|S|-1\right) \prod_{e\in S}\frac{B_e-x_e}{1-x_e}\right)
\EEAS
$Z$ can be seen as an expectation of $P$ where the $B_e$ are independent
Bernoulli random variables with parameter $x_e$. In particular expanding
$P$, we see that only the terms $(B_e-x_e)^2$ will contribute to its
expectation so that we get
\BEAS
Z &=& 1+\sum_{\emptyset\neq F\subset E}\prod_v \left( \frac{(-1)^{d_F(v)-1}(d_F(v)-1)}{\prod_{e\in \d v \cap F}(1-x_e)}\right)\prod_{e\in F}x_e(1-x_e)\\
&=&1+\sum_{\emptyset\neq F\subset E}(-1)^{V(F)}\prod_v (d_F(v)-1)\prod_{e\in F}\frac{x_e}{1-x_e},
\EEAS
where in the last claim, we used $\prod_v (-1)^{d_F(v)}=1$.
\end{proof}

The following lemma shows the last statement in Theorem \ref{th3}.
\begin{lemma}
We have
\BEAS
\ln Z = \Phi_G(z)-\Phi^B(\bx(z);z).
\EEAS
\end{lemma}
\begin{proof}
We first compute 
\BEAS
e^{\Phi^B(\bx;z)}&=& z^{-U^B(\bx)}\prod_v\left( 1-\sum_{e\in \d v}x_e\right)^{-(1-\sum_{e\in \d v}x_e)}\\
&&\prod_v\left(\prod_{e\in \d v} x_e^{-x_e/2}(1-x_e)^{(1-x_e)/2}\right)\\
&=& z^{\sum_e x_e}\prod_e x_e(1-x_e)^{1-x_e}\\
&&\prod_v \left( 1-\sum_{e\in \d v}x_e\right)^{-(1-\sum_{e\in \d v}x_e)}.
\EEAS
We now use the relation (\ref{eq:xe(z)}), to get:
\BEA
\label{eq:ePhi}e^{\Phi^B(\bx;z)}&=&\prod_e \frac{z}{x_e}\prod_v\left( 1-\sum_{e\in \d v}x_e\right)^{|\d v| -1}.
\EEA
We now compute
\BEAS
\frac{\prod_{v\in V} \mu_{\d v}(\bB_{\d
    v})}{\prod_{e\in E} \mu_e(B_e)} &=& \prod_v \left(1-\sum_{e\in \d v}x_e \right)^{1-\sum_{e \in \d v}B_e}\\
&&\prod_e (1-x_e)^{b_e-1},
\EEAS
again using (\ref{eq:xe(z)}), we obtain
\BEAS
\frac{\prod_{v\in V} \mu_{\d v}(\bB_{\d
    v})}{\prod_{e\in E} \mu_e(B_e)} &=& z^{\sum_e B_e}\prod_e \frac{x_e}{z}\\
&&\prod_v\left( 1-\sum_{e\in \d v}x_e\right)^{-|\d v|+1}.
\EEAS
Thanks to (\ref{eq:ePhi}), we get
\BEAS
\frac{\prod_{v\in V} \mu_{\d v}(\bB_{\d
    v})}{\prod_{e\in E} \mu_e(B_e)} &=& \frac{z^{\sum_e B_e}}{e^{\Phi^B(\bx;z)}},
\EEAS
hence summing over all matchings $\bB$, we obtain
\BEAS
Z = \frac{P_G(z)}{e^{\Phi^B(\bx;z)}}.
\EEAS
\end{proof}

\end{document}